\documentclass[aps,pra,twocolumn,superscriptaddress]{revtex4-1}
\usepackage{amsfonts,epsfig}
\usepackage[color]{lapdf}
\usepackage{hyperref}
\usepackage{latexsym}
\usepackage{amsmath}
\usepackage{amssymb}
\usepackage{amsfonts}
\usepackage{amsthm}
\usepackage{mathrsfs}
\usepackage{natbib}
\usepackage{color,verbatim,graphics}
\usepackage{psfrag}
\usepackage{verbatim,cleveref}
\usepackage{graphicx}
\usepackage{wrapfig}
\usepackage{fullpage}
\usepackage{arrayjobx}
\input{Qcircuit}

\makeatletter
\newtheorem*{rep@theorem}{\rep@title}
\newcommand{\newreptheorem}[2]{%
\newenvironment{rep#1}[1]{%
 \def\rep@title{#2 \ref{##1}}%
 \begin{rep@theorem}}%
 {\end{rep@theorem}}}
\makeatother

\newcommand{\suppl}{appendix}
\newcommand{\bop}{\mathcal{B}}
\newcommand{\hil}{\mathcal{H}}

\newcommand{\projS}[1]{\ket{#1}\!\!\bra{#1}}
\newcommand{\proj}[2]{\ket{#1}\!\!\bra{#2}}

\newcommand{\braket}[2]{\left \langle #1 | #2 \right \rangle}

\newtheorem{definition}{Definition}

\newtheorem{theorem}{Theorem}

\newtheorem{lemma}{Lemma}

\newreptheorem{lemma}{Lemma}
\newreptheorem{theorem}{Theorem}

\newcommand{\indexset}[1]{\left[#1\right]}

\begin{document}

\title{Entanglement improves classical control}

\author{Nguyen Truong Duy}
\affiliation{Centre for Quantum Technologies, National University of Singapore, 3 Science Drive 2, Singapore 117543}
\affiliation{School of Computing, National University of Singapore, 13 Computing Drive, Singapore 117417}
\author{Matthew McKague}
\affiliation{Centre for Quantum Technologies, National University of Singapore, 3 Science Drive 2,Singapore 117543}
\affiliation{Department of Physics, University of Otago, Dunedin, New Zealand}
\author{Stephanie Wehner}
\affiliation{Centre for Quantum Technologies, National University of Singapore, 3 Science Drive 2,Singapore 117543}
\affiliation{School of Computing, National University of Singapore, 13 Computing Drive, Singapore 117417}

\begin{abstract}
	Electronic devices all around us contain classical control circuits. Such circuits consist of a network of controllers which can read and write signals to wires
	of the circuit with the goal to minimize the cost function of the circuit's output signal. 
	Here, we propose the use of shared entanglement between controllers as a resource to improve the performance of otherwise purely classical control circuits. 
	We study a well-known example from the classical control literature and demonstrate that allowing two controllers to share entanglement improves their ability to control. 
	More precisely, we exhibit a family of circuits in which the the cost function using entanglement stays constant, but the minimal cost function without entanglement
	grows arbitrarily large. This demonstrates that entanglement can be a powerful resource in a classical control circuit.
\end{abstract}

\maketitle


Control theory has been used with great success to enable the control of \emph{quantum} systems. Whether it comes to 
battling decoherence~\cite{viola:dynamical}, or ensuring error-correction~\cite{qControlECC}, it is certainly fair to say that quantum control forms 
a necessary tool to build quantum devices of any kind (see e.g.~\cite{dong:quantumControl} for a recent survey). Indeed, quantum control has even been employed to study
fundamental aspects of quantum mechanics itself~\cite{delayedChoice}. In all of these works, the 
aim was to extend techniques and methods from classical control theory to gain a handle on \emph{quantum} systems, which forms a daunting task due the possibility 
of coherent control in superposition and entanglement.

Here, however, we raise the opposite question, namely whether quantum effects can improve \emph{classical}~\footnote{Throughout, the term ''classical'' refers to non-quantum
and not a particular form of control as in~\cite{witsenhausen:information}.} control itself. More precisely, we consider a classical control circuit in which all signals are classical and the overall aim of the circuit is to solve a \emph{classical} control problem. The only quantum ingredient will be that we allow the controllers in the circuit to share entanglement and to perform quantum measurements on the entangled state. 
To see whether such entanglement can at all be useful, we study the simple toy circuit depicted in Figure~\ref{fig:ProblemCircuit} which was proposed by Witsenhausen~\cite{witsenhausen} to study aspects of classical control and is well-known under the name ``Witsenhausen's counterexample'' (see e.g.~\cite{ho:witsenReview,sahai} for surveys). 
In its original form, Witsenhausen's circuit~\cite{witsenhausen} dealt with continuous signals but it was later discretized~\cite{ho:discrete,papadimitriou}, which is the approach we will 
follow here.
The objective of this circuit is to take an input signal on a line and subsequently damp it down in as efficient a manner as possible.  
We can accomplish this using two controllers, $c_1$ and $c_2$, which can each add a signal to the line. The first controller has perfect information, but its use incurs a cost. 
The second controller has imperfect information, but there is no penalty for its use.  
The problem is made richer by giving the the controllers $c_{1}$ and $c_{2}$ access to additional resources summarized by $\rho$. Classically, $\rho$ is a shared random variable, 
but quantumly $\rho$ can also be a shared entangled state.

\begin{figure*}[!]
\centering
\includegraphics[scale=0.9]{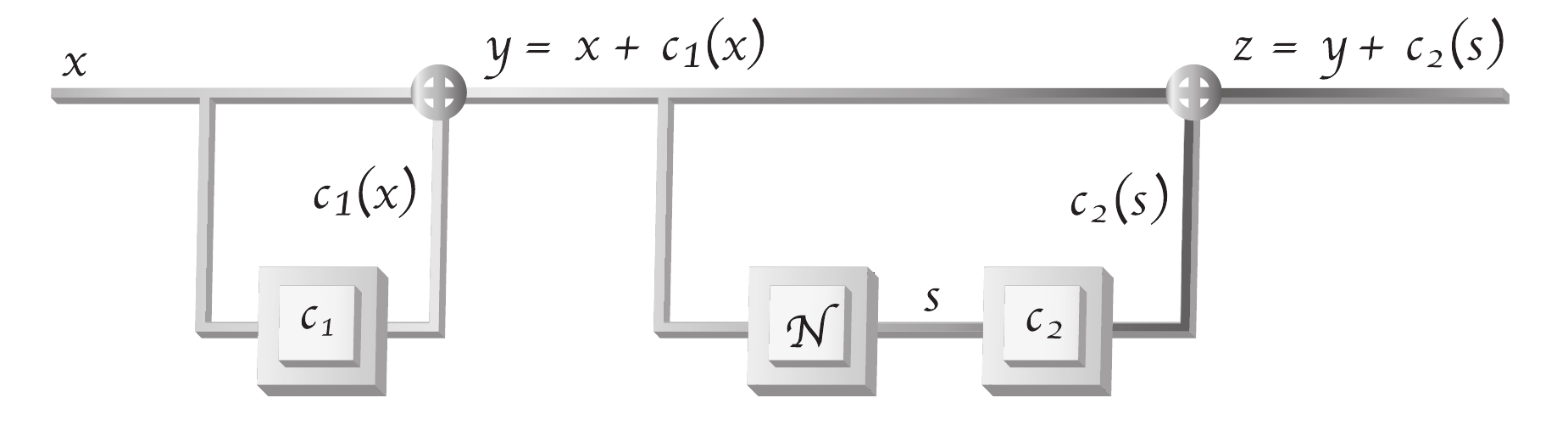}



\caption{Discrete Witsenhausen Circuit - The input to the circuit is a random integer $x \in \mathbb{Z}$ chosen according to a probability distribution $P_X$.
The input is fed into a controller $c_{1}$ which outputs an integer. The output from $c_{1}$ is then added to the line, giving $y = x + c_{1}(x)$. 
At this point the second controller, $c_{2}$, receives the signal from the line. However, the signal is received via a classical noisy channel $\mathcal{N}$, characterized by the 
probability $\mathcal{N}(s|y)$ of outputting the symbol $s \in \mathbb{Z}$ given input $y$. 
The second controller $c_{2}$ then again outputs an integer which is added to the line, giving the final output of the circuit $z = y + c_{2}(s)$.}
\label{fig:ProblemCircuit}
\end{figure*}

The performance of the controllers is evaluated using a cost function which grows with the remaining signal left on the line following the controllers.
More precisely, in terms of the variables defined in Figure~\ref{fig:ProblemCircuit}
\begin{align}
\label{eq:costFunc}
&	C^{\mathcal{N}, P_X, k} (c_{1}, c_{2}, \rho) = \nonumber \\
& \, \, \, \, \,\, \, \, \, \underset{P_X, \mathcal{N}}{\mathbb{E}}\left[k(c_{1}(x))^{2} + (x + c_{1}(x) + c_{2}(s))^{2}\right]
\end{align}
where $k \in \mathbb{R}^{+}$ is a parameter. A specific instance of the discrete Witsenhausen circuit is characterized by the choice of input distribution $P_X$,
the channel $\mathcal{N}$ and the parameter $k$ defining the cost function, i.e., by $(\mathcal{N}, P_X, k)$.  
Meanwhile, a specific control strategy is described by $(c_{1}, c_{2}, \rho)$.

We consider three classes $\Omega$ of strategies. The first two are \emph{classical} strategies. 
These can be deterministric strategies $D$, 
where $c_{1}$ and $c_{2}$ are functions $\mathbb{Z} \rightarrow \mathbb{Z}$ and $\rho$ may be taken to be $0$. 
Classically, we also allow strategies $SR$ that employ some shared randomness $\rho$. In this case we have that for all $(c_{1}, c_{2}, \rho) \in SR$, $c_{1}$, $c_{2}$ are functions $\mathbb{Z} \times R \rightarrow \mathbb{Z}$ where $\rho \in R$ is the shared random variable. Finally, we allow for \emph{quantum} strategies, where we use $SE$ to denote this 
class of strategies with shared entanglement $\rho$. Note that the input and output of the controllers remain purely classical. For any $(c_{1}, c_{2}, \rho) \in SE$, 
$\rho \in \bop(\hil_1\otimes \hil_2)$ is a quantum state shared between controller 1 holding system $\hil_1$ and controller 2 holding system $\hil_2$.
Without loss of generality, $c_{1}$ performs a measurement on $\hil_1$ indexed by its input $x$, and outputs the measurement outcome $m = c_1(x)$ on the wire. 
Formally, this means that the strategy
of $c_1$ is defined by measurement operators $\{\{\Pi_{x,m}^{1} \geq 0\}_{m}\}_x$ where for each $x$
we have $\sum_{m} \Pi_{x,m}^{1} = I_{H_{1}}$. 
Similarly, the second controller $c_2$ performs a measurement $\{\Pi_{s,\ell}^{2}\}_{\ell}$ on $\hil_2$ indexed by its input $s$ and outputs the outcome $\ell = c_2(s)$ on the wire. 
Given some instance $(\mathcal{N}, P_X, k)$ we define the \emph{minimum cost} for a class of strategies $\Omega \in \{D,SR,SE\}$ to be
\begin{equation}
C_{\Omega}^{\mathcal{N}, P_X, k} = \min_{(c_{1}, c_{2}, \rho) \in \Omega} \text{C}^{\mathcal{N}, P_X, k} (c_{1}, c_{2}, \rho)\ .
\end{equation}
%

\section{Result}

Here, we prove that a strategy using entanglement outperforms any classical strategy for the circuit given in Figure~\ref{fig:ProblemCircuit} and hence entanglement improves classical 
control. More precisely, we show that for the generalized discrete Witsenhausen circuit, there exist a classical channel $\mathcal{N}_{0}$, a probability distribution $P_X^{0}$ and a constant $k$ such that
\begin{equation}
C^{\mathcal{N}_{0}, P_X^{0}, k}_{SE} < C^{\mathcal{N}_{0}, P_X^{0}, k}_{SR}\ .
\end{equation}
We prove this theorem by exhibiting a particular strategy in $SE$ and a one-parameter family of instances $(\mathcal{N}_{t}, P_X^{t},k)$ such that 
$C_{SE}^{\mathcal{N}_{t}, P_X^{t}, k}$ remains constant for all $t$, but $C_{SR}^{\mathcal{N}_{t}, P_X^{t}, k}$ increases with $t$ without bound. 
Hence for sufficiently large $t$ our quantum strategy is better than the minimum cost for $SR$. What's more, it can be made arbitrarily large, showing a strong
separation between classical and quantum strategies.


\section{Methods}

Let us now give an overview of our proof - full details can be found in the \suppl. First, let us specify the parameters of the Witsenhausen instance. 
We will use a channel $\mathcal{N}_t = \mathcal{N} \circ \epsilon_t$, where $\mathcal{N}: [q]\times[d] \rightarrow ([q] \times [d])^{\times 2}$ with the notation $[q] = \{0,\ldots,q-1\}$. Such a channel for e.g.\ $(q,d) = (6,4)$ was studied in~\cite{winter}, where the task was to transmit messages $m \in [q]$ 
over $\mathcal{N}$ with \emph{zero} 
error using an encoder $\mathcal{E}$ and decoder $\mathcal{D}$. 
Here, we will make use of two properties of $\mathcal{N}$ established in~\cite{winter}: 
\begin{enumerate}
\item For any classical encoder and decoder, the maximal number of symbols that can be transmitted over $\mathcal{N}$ with \emph{zero} error is $q-1$.
\item For a quantum encoder and decoder, there exists a strategy to send $q$ symbols with zero error. More precisely, 
there exists a quantum state $\ket{\Psi}$ shared between the encoder and 
decoder, as well as measurements by the encoder depending on its input 
$m \in [q]$ giving outcome $j$, and a measurement 
by the decoder depending on its input $s = \mathcal{N}((m,j))$
such that the decoder recovers $(m,j)$ with probability $1$ for all $m \in [q]$. 
\end{enumerate}
The purpose of the channel $\epsilon_t$ is to map integers $y = mt + c_1(x)$ to inputs $(m,j)$ accepted by the channel $\mathcal{N}$. This is accomplished by 
letting
$\epsilon_t: \mathbb{Z} \rightarrow \indexset{q} \times \indexset{d}$ for $t \geq d$:
\begin{itemize}
	\item If $x = at + b$ for some $a \in \left[q\right]$ and $b \in \left[d\right]$ then $\epsilon_{t}$ outputs $(a,b)$.
	\item Otherwise $\epsilon_{t}$ outputs $(a, b)$ chosen uniformly at random from $\left[q\right] \times \left[d\right]$.
\end{itemize}
It is not difficult to see that the maximum number of symbols that can be transmitted over $\mathcal{N}_t$ with \emph{zero} error using a purely classical encoder and decoder remains $q-1$. The parameter $k > 0$ can be chosen arbitrarily, independent of $t$ and $\mathcal{N}$.

The input distribution we will use in our Witsenhausen instance is $P_X^t(x) = 1/q$ if $x = mt$ and $P_X^t(x) = 0$ otherwise. Intuitively, one can thus think of the input as being $m$, amplified by $t$. Our argument is identical for any other distribution on $m$.

For concreteness the reader may assume $(q,d)=(6,4)$ and $k=1$.

\subsubsection{Quantum cost function}
Let us first show that there exists a quantum strategy for which the cost function is independent of $t$. This is an easy consequence of the quantum strategy for zero error coding given in~\cite{winter} adapted to our setting.

Let controllers one and two share the state $\rho = \projS{\Psi}$ from the zero error coding scheme above. Upon receiving $x = mt$, controller $c_{1}$ 
first applies $\epsilon_t(x) = (m,0)$ himself to find $m$. He then uses the measurement for $m$ performed by the encoder $\mathcal{E}$ in~\cite{winter} on his part of $\ket{\Psi}$ to obtain $j$. He outputs 
$c_1^t(x) = j$ to be added to the signal which is then $y = x + j = mt+j$. 

Note that controller $c_{2}$ receives 
output $s = \{(m,j), (m^{\prime}, j^{\prime})\}$ from the 
channel $\mathcal{N}_{t}$.
He performs the measurement by the decoder of~\cite{winter} to obtain $(m,j)$. 
He then outputs $c_2^t(s) = -mt - j$ to be added to the signal, which is 
gives $z = y +  c_2^t(s) = 0$.

Since the final signal on the line is $z = mt + j - mt - j = 0$, the cost for this strategy is determined only by the term $k(c_{1}^t(x))^{2}$. Using the fact that $c_1^t(x) \leq d - 1$ we can show that for our strategy
\begin{align}
C^{\mathcal{N}_{t}, P_X^{t}, k}(c_{1}^t, c_{2}^t, \projS{\Psi}) 
& = k\sum_{m \in \left[q\right]} P_X(mt) (c_{1}^t(mt))^{2}\\
& \leq k d^2\ .
\end{align}

\subsubsection{Classical cost function}
The challenging part is to show that the classical cost function can be made arbitrarily large with $t$. We accomplish this in a series of lemmas, which together yield the theorem below. The main intuition behind our argument is the observation that when $t$ grows very large, but the cost function is still bounded, then controller two must have obtained a reasonable estimate for any possible $m \in [q]$. 
We then argue by contradiction and show that if 
the cost function is bounded, then controller one and controller two can transform their control strategy into a strategy for sending more then $q-1$ 
symbols over the channel $\mathcal{N}$ with \emph{zero} error, which is a contradiction.

\begin{theorem}
	Let $(\mathcal{N}_{t}, P_X^{t}, k)$ be a Witsenhausen instance as outlined above. Then for any $M \in \mathbb{R}$ there exists a $t_{0} \geq d$ such that for all $t \geq t_{0}$
	\begin{equation}
	C_{SR}^{\mathcal{N}_{t}, P_X^{t}, k} > M.
\end{equation}
\end{theorem}

\begin{proof}
We first show that we can restrict ourselves to deterministic strategies (Lemma~\ref{lemmaRandom}). Hence it suffices to show that 
\begin{equation}
		C_{D}^{\mathcal{N}_{t}, P_X^{t}, k} > M
\end{equation}
for large enough $t$. To this end, let $(c_{1}^{t}, c_{2}^{t},0)$ be a family of deterministic strategies which achieve a cost of $C_{D}^{\mathcal{N}_{t}, P_X^{t}, k}$. By assuming that 
\begin{equation}
\label{eq:boundforcontradictionMain}
		C_{D}^{\mathcal{N}_{t}, P_X^{t}, k} \leq M
\end{equation}
for all $t$, we will derive a contradiction. Using the definition of the cost function, this means that
\begin{equation}
	C_{D}^{\mathcal{N}_{t}, P_X^{t}, k}(c_{1}^{t}, c_{2}^{t},0) = 
	\underset{P_X, \mathcal{N}_t}{\mathbb{E}} \left[ 
		k c_{1}^{t}(x) + z^{2}
	\right] \leq M\ ,
\end{equation}
where $z = x + c_1^t(x) + c_{2}^{t}(s)$.

Since the expectation is taken over a finite set of events of non-zero probability we can show that $|c_{1}^t(x)|$ and $|z|$ are uniformly bounded for all possible inputs $x$ (see Lemma~\ref{boundC1}). That is, 
\begin{eqnarray}
	|c_{1}^t(x)| & \leq& M_{X}\ , \\
	|z| & \leq & M_{Z}\ ,
\end{eqnarray}
for some $M_{X}$ and $M_{Z}$ which are independent of $t$ since they can be defined in terms of the parameters $M$, $k$, $P_{X}^{t\, min}$ and $P_{Z}^{t\, min}$ which are all independent of $t$. (The distributions $P_{X}^{t}$ and $P_{Z}^{t}$ depend on $t$, but the \emph{minimal} probability in these distributions does not.)  

Then, for large enough $t$, $|c_{1}^{t}(x)|$ is small compared with the inputs which are of the form $x=mt$. 
Dividing by $t$ to make this more apparent, the signal $y = x + c_1^t(x)$ after the first controller satisfies
\begin{eqnarray}
	\frac{y}{t} & = &  m + \frac{c_{1}^{t}(mt)}{t} \\
	& \leq & m + \frac{M_{X}}{t} \\
	& \approx & m 
\end{eqnarray}
for large $t$ since $M_{X}/t \rightarrow 0$ as $t \rightarrow \infty$. Now the final output satisfies
\begin{eqnarray}
	\frac{z}{t} & = & m + \frac{c_{1}^{t}(mt)}{t} + \frac{c_{2}^{t}(s)}{t} \\
	& \approx &  m + \frac{c_{2}^{t}(s)}{t} \\
	& \approx & 0
\end{eqnarray}
for large $t$. The second line follows from the fact that $c_{1}^t(x)$ is uniformly bounded for all possible $x$ (Lemma~\ref{boundC1}), and the 
final line follows from the fact that also $|z|$ is uniformly bounded. Hence, we have that $z/t \rightarrow 0$ as $t \rightarrow \infty$. 
Lemma~\ref{lemma:c2approximatesmt} makes this argument formal.  In particular, when $t_0 = 2(M_{X} + M_{Y})+ 1$, for any $t \geq t_{0}$ we have
\begin{align}\label{eq:estimateMain}
	\left|\frac{x}{t} + \frac{c_2^{t}(s)}{t} \right| < \frac{1}{2}
\end{align}
for all possible $(x,s)$  such that $P_{X,S}(x,s) > 0$ (in particular, for all $x = mt$ with $m \in [q]$).  

In our concrete example with $(q,d) = (6,4)$ and $k=1$ we can bound $M_X$ and $M_Z$ from above by $\sqrt{6M}$ and $\sqrt{54M}$, respectively.  In this case, $t_0$ is bounded above by $20\sqrt{M} + 1$.
	
From now on let $t \geq t_{0}$.  Note that~\eqref{eq:estimateMain} means that $c_2^t(s)/t$ forms a good estimate for $m = x/t$. This allows us 
to construct a zero error encoding scheme for the channel $\mathcal{N}_{t}$.
We use message set $\mathcal{M} = [q]$ and let 
\begin{equation}
	\eta = \frac{-c_2^{t}(s)}{t}\ .
\end{equation}
We then define the encoding scheme $\mathcal{E}$ and decoding scheme $\mathcal{D}$ as
\begin{itemize}
	\item $\mathcal{E}(m) = mt + c_{1}^{t}(mt)$.
	\item $\mathcal{D}(s)$ is given by rounding off $\eta$ to the nearest integer.
\end{itemize}
Since $|m - \eta| < 1/2$ by~\eqref{eq:estimateMain}, the nearest integer to $\eta$ is always $m$ and $\mathcal{D}$ always decodes correctly.  Here the fact that we have used a particular distribution for the inputs does not matter.  Since we achieve zero probability of error the encoding strategy must work with probability 1 for every input with positive probability.  Put differently, our encoding strategy works with zero error for an alphabet of $q$ symbols.  This contradicts the fact that only $q-1$ symbols can be sent over $\mathcal{N}_t$ with zero error.
\end{proof}

\section{Discussion}

We demonstrated that entanglement improves classical control by studying a very simple circuit. It is an interesting open question to decide for which other control circuits
entanglement is useful, and whether one could maybe even characterize the set of circuits for which entanglement helps. 
We note that a number of existing tasks in quantum information could be converted into a problem of classical control. First, Bell inequalities~\cite{bell} expressed
in the language of non-local games (see e.g.~\cite{Bellsurvey} for a survey) can be turned into a control problem in which two non-communicating controllers, Alice and Bob, 
receive inputs and have to produce outputs trying to maximize their winning probability which in this context can be seen as the 
negative of the cost function. The circuit obtained this way is somewhat simpler than Figure~\ref{fig:ProblemCircuit} 
because there is no feedforward, but nevertheless allows the construction of a primitive 
classical circuit in which known results on Bell inequalities tell us that entanglement does help.
Indeed, one may see such non-local games as very restricted cases of the discrete team decision problem studied
in the context of Witsenhausen's counterexample~\cite{papadimitriou,radner:team}.

Another task that allows us to come up with a classical circuit in which entanglement can help is the one of sending information over a noisy-channel, 
where the encoder and the decoder can share entanglement. In this context, it is known that entanglement helps to increase the zero-error capacity of 
the channel~\cite{winter,winter2} and one can turn this problem into a less trivial control circuit by defining 
a cost function that introduces a penalty for incorrect decoding. Looking at Figure~\ref{fig:ProblemCircuit} it becomes clear that this circuit
is already somewhat related to Witsenhausen's counterexample, because transmitting information correctly from one controller to the other helps them
to minimize the cost function, underscoring the information-theoretic flavor of this circuit~\cite{mitterSahai:revisited,witsenhausen:information}
Indeed, Witsenshausen's counterexample can be cast in the language of information transmission~\cite{GroverSahaiPark} 
where minimizing one part of the cost function can be seen as minimizing the mean-square error of correctly transmitting information from 
one controller to the other. As such, the techniques
of~\cite{winter,winter2} are useful for us to establish one half of our proof, namely that the strategy using shared entanglement
leads to a cost function that is independent of the parameter $t$. However, Witsenhausen's counterexample has the 
additional twist that the use of controller 1 introduces an additional penalty term into the cost function 
that is normally not considered in quantum information theory.

Taking an extremely broad perspective, one could even take computational problems and define a cost function based on whether a particular
circuit correctly computes the solution of some function $f$ on a given input $x$. Here, measurement based computing~\cite{measurementBasedQC} 
could be seen as a possible control strategy to minimize the cost, using controllers sharing a quantum state. 

We emphasize, however, that our aim is not to take the many tasks of quantum information and turn them into specific instances of classical
control circuits. Rather, our goal is to raise the general question whether entanglement can improve classical control circuits, and 
we have shown that this is indeed possible by introducing entanglement into a well-studied and established circuit from the control literature.

Witsenhausen's counterexample was originally conceived to demonstrate that linear control laws are not always optimal~\cite{witsenhausen}. 
It has since been studied extensively to determine the computational difficulty of computing optimal classical control strategies~\cite{wu:witsenhausen,grover:vector,GroverSahaiPark},
and various heuristic algorithms are known~\cite{baglietto,lee,li}. Intriguingly, the optimal 
strategy for Witsenhausen's original circuit remains unknown (see e.g.~\cite{GroverSahaiPark}), 
and discrete versions~\cite{ho:discrete} have even been proven to be NP-complete~\cite{papadimitriou}. 
As such, it would be very interesting to determine the complexity of determining the optimal control strategy in the presence of entanglement.
Note that this can be seen as a relaxation of the original question because strategies involving entanglement include all classical strategies
as a special case, and it is hence conceivable that solving the quantum question could be easier than determining
the optimal classical strategy. This would give a method for obtaining bounds for classical strategies. In the setting of Bell inequalities discussed
above, it is known that for some non-local games, i.e.\ some specific control problems, known as XOR games, the optimal quantum strategy is 
easy to compute~\cite{watrous:nl,wehner06b} whereas finding the optimal classical strategy is hard. Yet, for other games the question is known
to be hard even in the presence of entanglement~\cite{julia:NP,thomas:NP}. The question of how difficult it is to compute the optimal \emph{quantum} strategy
for Witsenhausen's circuit is thus wide open.

\acknowledgments
We thank Andrew Doherty for valuable pointers to the quantum control literature.
This research was funded by the Ministry of Education (MOE), the National Research Foundation, Singapore, MOE Tier 3 grant MOE2012-T3-1-009, the University of Otago through a University of Otago Research Grant and the Performance Based Research Fund, and the Jack Dodd Centre for QuantumÊTechnology.
\bibliography{controlProblem.v2}

\begin{thebibliography}{29}%
\makeatletter
\providecommand \@ifxundefined [1]{%
 \@ifx{#1\undefined}
}%
\providecommand \@ifnum [1]{%
 \ifnum #1\expandafter \@firstoftwo
 \else \expandafter \@secondoftwo
 \fi
}%
\providecommand \@ifx [1]{%
 \ifx #1\expandafter \@firstoftwo
 \else \expandafter \@secondoftwo
 \fi
}%
\providecommand \natexlab [1]{#1}%
\providecommand \enquote  [1]{``#1''}%
\providecommand \bibnamefont  [1]{#1}%
\providecommand \bibfnamefont [1]{#1}%
\providecommand \citenamefont [1]{#1}%
\providecommand \href@noop [0]{\@secondoftwo}%
\providecommand \href [0]{\begingroup \@sanitize@url \@href}%
\providecommand \@href[1]{\@@startlink{#1}\@@href}%
\providecommand \@@href[1]{\endgroup#1\@@endlink}%
\providecommand \@sanitize@url [0]{\catcode `\\12\catcode `\$12\catcode
  `\&12\catcode `\#12\catcode `\^12\catcode `\_12\catcode `\%12\relax}%
\providecommand \@@startlink[1]{}%
\providecommand \@@endlink[0]{}%
\providecommand \url  [0]{\begingroup\@sanitize@url \@url }%
\providecommand \@url [1]{\endgroup\@href {#1}{\urlprefix }}%
\providecommand \urlprefix  [0]{URL }%
\providecommand \Eprint [0]{\href }%
\providecommand \doibase [0]{http://dx.doi.org/}%
\providecommand \selectlanguage [0]{\@gobble}%
\providecommand \bibinfo  [0]{\@secondoftwo}%
\providecommand \bibfield  [0]{\@secondoftwo}%
\providecommand \translation [1]{[#1]}%
\providecommand \BibitemOpen [0]{}%
\providecommand \bibitemStop [0]{}%
\providecommand \bibitemNoStop [0]{.\EOS\space}%
\providecommand \EOS [0]{\spacefactor3000\relax}%
\providecommand \BibitemShut  [1]{\csname bibitem#1\endcsname}%
\let\auto@bib@innerbib\@empty
\bibitem [{\citenamefont {Viola}\ and\ \citenamefont
  {Lloyd}(1998)}]{viola:dynamical}%
  \BibitemOpen
  \bibfield  {author} {\bibinfo {author} {\bibfnamefont {L.}~\bibnamefont
  {Viola}}\ and\ \bibinfo {author} {\bibfnamefont {S.}~\bibnamefont {Lloyd}},\
  }\href@noop {} {\bibfield  {journal} {\bibinfo  {journal} {Phys. Rev. A}\ ,\
  \bibinfo {pages} {2733}} (\bibinfo {year} {1998})}\BibitemShut {NoStop}%
\bibitem [{\citenamefont {Ahn}\ \emph {et~al.}(2002)\citenamefont {Ahn},
  \citenamefont {Doherty},\ and\ \citenamefont {Landahl}}]{qControlECC}%
  \BibitemOpen
  \bibfield  {author} {\bibinfo {author} {\bibfnamefont {C.}~\bibnamefont
  {Ahn}}, \bibinfo {author} {\bibfnamefont {A.~C.}\ \bibnamefont {Doherty}}, \
  and\ \bibinfo {author} {\bibfnamefont {A.~J.}\ \bibnamefont {Landahl}},\
  }\href@noop {} {\bibfield  {journal} {\bibinfo  {journal} {Phys. Rev. A}\
  }\textbf {\bibinfo {volume} {65}},\ \bibinfo {pages} {042301} (\bibinfo
  {year} {2002})}\BibitemShut {NoStop}%
\bibitem [{\citenamefont {Dong}\ and\ \citenamefont
  {Petersen}(2010)}]{dong:quantumControl}%
  \BibitemOpen
  \bibfield  {author} {\bibinfo {author} {\bibfnamefont {D.}~\bibnamefont
  {Dong}}\ and\ \bibinfo {author} {\bibfnamefont {I.~R.}\ \bibnamefont
  {Petersen}},\ }\href@noop {} {\bibfield  {journal} {\bibinfo  {journal}
  {{IET} {C}ontrol {T}heory and {A}pplications}\ }\textbf {\bibinfo {volume}
  {4}},\ \bibinfo {pages} {2651} (\bibinfo {year} {2010})}\BibitemShut
  {NoStop}%
\bibitem [{\citenamefont {Ionicioiu}\ and\ \citenamefont
  {Terno}(2011)}]{delayedChoice}%
  \BibitemOpen
  \bibfield  {author} {\bibinfo {author} {\bibfnamefont {R.}~\bibnamefont
  {Ionicioiu}}\ and\ \bibinfo {author} {\bibfnamefont {D.~R.}\ \bibnamefont
  {Terno}},\ }\href@noop {} {\bibfield  {journal} {\bibinfo  {journal} {Phys.
  Rev. Lett.}\ }\textbf {\bibinfo {volume} {107}},\ \bibinfo {pages} {230406}
  (\bibinfo {year} {2011})}\BibitemShut {NoStop}%
\bibitem [{Note1()}]{Note1}%
  \BibitemOpen
  \bibinfo {note} {Throughout, the term ''classical'' refers to non-quantum and
  not a particular form of control as in~\cite
  {witsenhausen:information}.}\BibitemShut {Stop}%
\bibitem [{\citenamefont {Witsenhausen}(1968)}]{witsenhausen}%
  \BibitemOpen
  \bibfield  {author} {\bibinfo {author} {\bibfnamefont {H.~S.}\ \bibnamefont
  {Witsenhausen}},\ }\href {\doibase 10.1137/0306011} {\bibfield  {journal}
  {\bibinfo  {journal} {Siam Journal on Control}\ }\textbf {\bibinfo {volume}
  {6}} (\bibinfo {year} {1968}),\ 10.1137/0306011}\BibitemShut {NoStop}%
\bibitem [{\citenamefont {Ho}(2008)}]{ho:witsenReview}%
  \BibitemOpen
  \bibfield  {author} {\bibinfo {author} {\bibfnamefont {Y.}~\bibnamefont
  {Ho}},\ }in\ \href@noop {} {\emph {\bibinfo {booktitle} {Proceedings of the
  47th IEEE Conference on decision and control ({CDC})}}}\ (\bibinfo {year}
  {2008})\ pp.\ \bibinfo {pages} {1611--1613}\BibitemShut {NoStop}%
\bibitem [{\citenamefont {Sahai}(2010)}]{sahai}%
  \BibitemOpen
  \bibfield  {author} {\bibinfo {author} {\bibfnamefont {A.}~\bibnamefont
  {Sahai}},\ }\href@noop {} {\bibfield  {journal} {\bibinfo  {journal} {IEEE
  Control Systems Magazine}\ ,\ \bibinfo {pages} {20}} (\bibinfo {year}
  {2010})}\BibitemShut {NoStop}%
\bibitem [{\citenamefont {Ho}\ and\ \citenamefont {Chang}(1980)}]{ho:discrete}%
  \BibitemOpen
  \bibfield  {author} {\bibinfo {author} {\bibfnamefont {Y.-C.}\ \bibnamefont
  {Ho}}\ and\ \bibinfo {author} {\bibfnamefont {T.}~\bibnamefont {Chang}},\
  }\href@noop {} {\bibfield  {journal} {\bibinfo  {journal} {{IEEE}
  {T}ransactions on {A}utomat. {C}ontr.}\ }\textbf {\bibinfo {volume} {25}},\
  \bibinfo {pages} {537} (\bibinfo {year} {1980})}\BibitemShut {NoStop}%
\bibitem [{\citenamefont {Papadimitriou}\ and\ \citenamefont
  {Tsitsiklis}(1986)}]{papadimitriou}%
  \BibitemOpen
  \bibfield  {author} {\bibinfo {author} {\bibfnamefont {C.~H.}\ \bibnamefont
  {Papadimitriou}}\ and\ \bibinfo {author} {\bibfnamefont {J.}~\bibnamefont
  {Tsitsiklis}},\ }\href {\doibase 10.1137/0324038} {\bibfield  {journal}
  {\bibinfo  {journal} {SIAM J. Control Optim.}\ }\textbf {\bibinfo {volume}
  {24}},\ \bibinfo {pages} {639} (\bibinfo {year} {1986})}\BibitemShut
  {NoStop}%
\bibitem [{\citenamefont {Cubitt}\ \emph {et~al.}(2010)\citenamefont {Cubitt},
  \citenamefont {Leung}, \citenamefont {Matthews},\ and\ \citenamefont
  {Winter}}]{winter}%
  \BibitemOpen
  \bibfield  {author} {\bibinfo {author} {\bibfnamefont {T.~S.}\ \bibnamefont
  {Cubitt}}, \bibinfo {author} {\bibfnamefont {D.}~\bibnamefont {Leung}},
  \bibinfo {author} {\bibfnamefont {W.}~\bibnamefont {Matthews}}, \ and\
  \bibinfo {author} {\bibfnamefont {A.}~\bibnamefont {Winter}},\ }\href@noop {}
  {\bibfield  {journal} {\bibinfo  {journal} {Physical Review Letters}\
  }\textbf {\bibinfo {volume} {104}},\ \bibinfo {pages} {230503} (\bibinfo
  {year} {2010})}\BibitemShut {NoStop}%
\bibitem [{\citenamefont {Bell}(1965)}]{bell}%
  \BibitemOpen
  \bibfield  {author} {\bibinfo {author} {\bibfnamefont {J.~S.}\ \bibnamefont
  {Bell}},\ }\href@noop {} {\ \textbf {\bibinfo {volume} {1}},\ \bibinfo
  {pages} {195} (\bibinfo {year} {1965})}\BibitemShut {NoStop}%
\bibitem [{\citenamefont {Brunner}\ \emph {et~al.}(2013)\citenamefont
  {Brunner}, \citenamefont {Cavalcanti}, \citenamefont {Pironio}, \citenamefont
  {Scarani},\ and\ \citenamefont {Wehner}}]{Bellsurvey}%
  \BibitemOpen
  \bibfield  {author} {\bibinfo {author} {\bibfnamefont {N.}~\bibnamefont
  {Brunner}}, \bibinfo {author} {\bibfnamefont {D.}~\bibnamefont {Cavalcanti}},
  \bibinfo {author} {\bibfnamefont {S.}~\bibnamefont {Pironio}}, \bibinfo
  {author} {\bibfnamefont {V.}~\bibnamefont {Scarani}}, \ and\ \bibinfo
  {author} {\bibfnamefont {S.}~\bibnamefont {Wehner}},\ }\href@noop {}
  {\enquote {\bibinfo {title} {Bell nonlocality},}\ } (\bibinfo {year}
  {2013}),\ \bibinfo {note} {arXiv:1303.2849}\BibitemShut {NoStop}%
\bibitem [{\citenamefont {Rander}(1962)}]{radner:team}%
  \BibitemOpen
  \bibfield  {author} {\bibinfo {author} {\bibfnamefont {R.}~\bibnamefont
  {Rander}},\ }\href@noop {} {\bibfield  {journal} {\bibinfo  {journal} {Annals
  of Mathematical Statistics}\ } (\bibinfo {year} {1962})}\BibitemShut
  {NoStop}%
\bibitem [{\citenamefont {Cubitt}\ \emph {et~al.}(2011)\citenamefont {Cubitt},
  \citenamefont {Leung}, \citenamefont {Matthews},\ and\ \citenamefont
  {Winter}}]{winter2}%
  \BibitemOpen
  \bibfield  {author} {\bibinfo {author} {\bibfnamefont {T.~S.}\ \bibnamefont
  {Cubitt}}, \bibinfo {author} {\bibfnamefont {D.}~\bibnamefont {Leung}},
  \bibinfo {author} {\bibfnamefont {W.}~\bibnamefont {Matthews}}, \ and\
  \bibinfo {author} {\bibfnamefont {A.}~\bibnamefont {Winter}},\ }\href
  {\doibase 10.1109/TIT.2011.2159047} {\bibfield  {journal} {\bibinfo
  {journal} {IEEE Trans. Inf. Theor.}\ }\textbf {\bibinfo {volume} {57}},\
  \bibinfo {pages} {5509} (\bibinfo {year} {2011})}\BibitemShut {NoStop}%
\bibitem [{\citenamefont {Mitter}\ and\ \citenamefont
  {Sahai}(1999)}]{mitterSahai:revisited}%
  \BibitemOpen
  \bibfield  {author} {\bibinfo {author} {\bibfnamefont {S.}~\bibnamefont
  {Mitter}}\ and\ \bibinfo {author} {\bibfnamefont {A.}~\bibnamefont {Sahai}},\
  }\href@noop {} {\emph {\bibinfo {title} {Information and control:
  Witsenhausen revisited}}},\ Lecture notes in control and information
  sciences\ (\bibinfo  {publisher} {Springer},\ \bibinfo {year} {1999})\ pp.\
  \bibinfo {pages} {281--293}\BibitemShut {NoStop}%
\bibitem [{\citenamefont {Witsenhausen}(1971)}]{witsenhausen:information}%
  \BibitemOpen
  \bibfield  {author} {\bibinfo {author} {\bibfnamefont {H.~S.}\ \bibnamefont
  {Witsenhausen}},\ }in\ \href@noop {} {\emph {\bibinfo {booktitle}
  {Proceedings of the IEEE}}},\ Vol.~\bibinfo {volume} {59}\ (\bibinfo {year}
  {1971})\ pp.\ \bibinfo {pages} {1557--1566}\BibitemShut {NoStop}%
\bibitem [{\citenamefont {Grover}\ \emph {et~al.}(2009)\citenamefont {Grover},
  \citenamefont {Sahai},\ and\ \citenamefont {Park}}]{GroverSahaiPark}%
  \BibitemOpen
  \bibfield  {author} {\bibinfo {author} {\bibfnamefont {P.}~\bibnamefont
  {Grover}}, \bibinfo {author} {\bibfnamefont {A.}~\bibnamefont {Sahai}}, \
  and\ \bibinfo {author} {\bibfnamefont {S.~Y.}\ \bibnamefont {Park}},\ }in\
  \href@noop {} {\emph {\bibinfo {booktitle} {Proceedings of the 7th
  international conference on Modeling and Optimization in Mobile, Ad Hoc, and
  Wireless Networks}}}\ (\bibinfo  {publisher} {IEEE},\ \bibinfo {year}
  {2009})\ pp.\ \bibinfo {pages} {604--613}\BibitemShut {NoStop}%
\bibitem [{\citenamefont {Raussendorf}\ \emph {et~al.}(2003)\citenamefont
  {Raussendorf}, \citenamefont {Browne},\ and\ \citenamefont
  {Briegel}}]{measurementBasedQC}%
  \BibitemOpen
  \bibfield  {author} {\bibinfo {author} {\bibfnamefont {R.}~\bibnamefont
  {Raussendorf}}, \bibinfo {author} {\bibfnamefont {D.~E.}\ \bibnamefont
  {Browne}}, \ and\ \bibinfo {author} {\bibfnamefont {H.~J.}\ \bibnamefont
  {Briegel}},\ }\href@noop {} {\bibfield  {journal} {\bibinfo  {journal} {Phys.
  Rev. A}\ }\textbf {\bibinfo {volume} {68}},\ \bibinfo {pages} {022312}
  (\bibinfo {year} {2003})}\BibitemShut {NoStop}%
\bibitem [{\citenamefont {Wu}\ and\ \citenamefont
  {Verd{\'u}}(2011)}]{wu:witsenhausen}%
  \BibitemOpen
  \bibfield  {author} {\bibinfo {author} {\bibfnamefont {Y.}~\bibnamefont
  {Wu}}\ and\ \bibinfo {author} {\bibfnamefont {S.}~\bibnamefont {Verd{\'u}}},\
  }in\ \href@noop {} {\emph {\bibinfo {booktitle} {Proceedings of the 50th
  {IEEE} {C}onference on {D}ecision and {C}ontrol ({CDC})}}}\ (\bibinfo {year}
  {2011})\BibitemShut {NoStop}%
\bibitem [{\citenamefont {Gover}\ and\ \citenamefont
  {Sahai}(2010)}]{grover:vector}%
  \BibitemOpen
  \bibfield  {author} {\bibinfo {author} {\bibfnamefont {P.}~\bibnamefont
  {Gover}}\ and\ \bibinfo {author} {\bibfnamefont {A.}~\bibnamefont {Sahai}},\
  }\href@noop {} {\bibfield  {journal} {\bibinfo  {journal} {Int. J. Syst.
  Control Commun. (Special Issue on Information Processing and decision making
  in distributed control systems)}\ }\textbf {\bibinfo {volume} {2}},\ \bibinfo
  {pages} {197} (\bibinfo {year} {2010})}\BibitemShut {NoStop}%
\bibitem [{\citenamefont {Baglietto}\ \emph {et~al.}(2001)\citenamefont
  {Baglietto}, \citenamefont {Parisini},\ and\ \citenamefont
  {Zoppoli}}]{baglietto}%
  \BibitemOpen
  \bibfield  {author} {\bibinfo {author} {\bibfnamefont {M.}~\bibnamefont
  {Baglietto}}, \bibinfo {author} {\bibfnamefont {T.}~\bibnamefont {Parisini}},
  \ and\ \bibinfo {author} {\bibfnamefont {R.}~\bibnamefont {Zoppoli}},\
  }\href@noop {} {\bibfield  {journal} {\bibinfo  {journal} {IEEE
  {T}ransactions on {A}utomat. {C}ontr.}\ }\textbf {\bibinfo {volume} {46}},\
  \bibinfo {pages} {1471} (\bibinfo {year} {2001})}\BibitemShut {NoStop}%
\bibitem [{\citenamefont {Lee}\ \emph {et~al.}(2001)\citenamefont {Lee},
  \citenamefont {Lau},\ and\ \citenamefont {Ho}}]{lee}%
  \BibitemOpen
  \bibfield  {author} {\bibinfo {author} {\bibfnamefont {J.~T.}\ \bibnamefont
  {Lee}}, \bibinfo {author} {\bibfnamefont {E.}~\bibnamefont {Lau}}, \ and\
  \bibinfo {author} {\bibfnamefont {Y.}~\bibnamefont {Ho}},\ }\href@noop {}
  {\bibfield  {journal} {\bibinfo  {journal} {IEEE {T}ransactions on {A}utomat.
  {C}ontr.}\ }\textbf {\bibinfo {volume} {46}},\ \bibinfo {pages} {382}
  (\bibinfo {year} {2001})}\BibitemShut {NoStop}%
\bibitem [{\citenamefont {Li}\ \emph {et~al.}(2009)\citenamefont {Li},
  \citenamefont {Marden},\ and\ \citenamefont {Shamma}}]{li}%
  \BibitemOpen
  \bibfield  {author} {\bibinfo {author} {\bibfnamefont {N.}~\bibnamefont
  {Li}}, \bibinfo {author} {\bibfnamefont {J.}~\bibnamefont {Marden}}, \ and\
  \bibinfo {author} {\bibfnamefont {J.~S.}\ \bibnamefont {Shamma}},\ }in\
  \href@noop {} {\emph {\bibinfo {booktitle} {Proceedings of the 48th IEEE
  {C}onference on {D}ecision and {C}ontrol ({CDC})}}}\ (\bibinfo {year}
  {2009})\ pp.\ \bibinfo {pages} {157--162}\BibitemShut {NoStop}%
\bibitem [{\citenamefont {Cleve}\ \emph {et~al.}(2004)\citenamefont {Cleve},
  \citenamefont {H{\o}yer}, \citenamefont {Toner},\ and\ \citenamefont
  {Watrous}}]{watrous:nl}%
  \BibitemOpen
  \bibfield  {author} {\bibinfo {author} {\bibfnamefont {R.}~\bibnamefont
  {Cleve}}, \bibinfo {author} {\bibfnamefont {P.}~\bibnamefont {H{\o}yer}},
  \bibinfo {author} {\bibfnamefont {B.}~\bibnamefont {Toner}}, \ and\ \bibinfo
  {author} {\bibfnamefont {J.}~\bibnamefont {Watrous}},\ }in\ \href@noop {}
  {\emph {\bibinfo {booktitle} {Proceedings of 19th IEEE Conference on
  Computational Complexity}}}\ (\bibinfo {year} {2004})\ pp.\ \bibinfo {pages}
  {236--249}\BibitemShut {NoStop}%
\bibitem [{\citenamefont {Wehner}(2006)}]{wehner06b}%
  \BibitemOpen
  \bibfield  {author} {\bibinfo {author} {\bibfnamefont {S.}~\bibnamefont
  {Wehner}},\ }\href@noop {} {\bibfield  {journal} {\bibinfo  {journal}
  {Physical Review A}\ }\textbf {\bibinfo {volume} {73}},\ \bibinfo {pages}
  {022110} (\bibinfo {year} {2006})}\BibitemShut {NoStop}%
\bibitem [{\citenamefont {Kempe}\ \emph {et~al.}(2007)\citenamefont {Kempe},
  \citenamefont {Kobayashi}, \citenamefont {Matsumoto}, \citenamefont {Toner},\
  and\ \citenamefont {Vidick}}]{julia:NP}%
  \BibitemOpen
  \bibfield  {author} {\bibinfo {author} {\bibfnamefont {J.}~\bibnamefont
  {Kempe}}, \bibinfo {author} {\bibfnamefont {H.}~\bibnamefont {Kobayashi}},
  \bibinfo {author} {\bibfnamefont {K.}~\bibnamefont {Matsumoto}}, \bibinfo
  {author} {\bibfnamefont {B.}~\bibnamefont {Toner}}, \ and\ \bibinfo {author}
  {\bibfnamefont {T.}~\bibnamefont {Vidick}},\ }\href@noop {} {\enquote
  {\bibinfo {title} {Entangled games are hard to approximate},}\ } (\bibinfo
  {year} {2007}),\ \bibinfo {note} {arXiv:0704.2903}\BibitemShut {NoStop}%
\bibitem [{\citenamefont {Vidick}(2013)}]{thomas:NP}%
  \BibitemOpen
  \bibfield  {author} {\bibinfo {author} {\bibfnamefont {T.}~\bibnamefont
  {Vidick}},\ }\href@noop {} {\enquote {\bibinfo {title} {Three-player
  entangled xor games are np-hard to approximate},}\ } (\bibinfo {year}
  {2013}),\ \bibinfo {note} {arXiv:1302.1242}\BibitemShut {NoStop}%
\bibitem [{Note2()}]{Note2}%
  \BibitemOpen
  \bibinfo {note} {This channel is different from one used by Cubbit et al.\
  \cite {winter} but it has the same confusability graph and deterministic
  capacity. As well, the same quantum encoding/decoding strategy also
  works.}\BibitemShut {Stop}%
\end{thebibliography}%

\newpage
\appendix
\pagebreak 

\section{Overview}
In this appendix we provide the technical details of our claims. 
Our goal is thereby to produce a family of Witsenhausen instances, $(\mathcal{N}_{t}, P_X^{t}, k)$ parameterized by $t > d$ for some constant $d$ such that
\begin{eqnarray}\label{eq:constantcost}
	C_{SE}^{\mathcal{N}_{t}, P_X^{t}, k} & \leq & C\ ,
\end{eqnarray}
where $C = kd^2$ is independent of $t$, but
\begin{eqnarray}
\label{eq:unbound}\label{eq:divergingcost}
	\lim_{t \rightarrow \infty} C_{SR}^{\mathcal{N}_{t}, P_X^{t}, k} & = & \infty
\end{eqnarray}
That is to say, for large enough $t$ the quantum strategy beats the classical bound, and that the difference can be made arbitrarily large as $t$ increases.

The  construction of our family of instances of the Witsenhausen problem is based on a classical channel $\mathcal{N}$ used by Cubitt et al. \cite{winter} and \cite{winter2}. 
In Section~\ref{sec:prelim} we recall their construction for completeness. 
In particular~\cite{winter,winter2} showed that the zero-error capacity of $\mathcal{N}$ without entanglement is strictly smaller than the entanglement-assisted capacity, and we will employ some
of their techniques here to prove the quantum part of our task.

In Section~\ref{sec:witsenhausen}, we then give the construction for our family of instances of the Witsenhausen problem based on $\mathcal{N}$. 
This is done by choosing a suitable encoding from the inputs in our circuit to inputs acceptable by $\mathcal{N}$. 
This encoding depends on the parameter $t$. The combined encoding composed with $\mathcal{N}$ gives us our noisy channels $\mathcal{N}_t$ to be used in the Witsenhausen problem.  
We also define a probability distribution $P_X^{t}$, depending on $t$ where the largest value with non-zero probability increases with $t$.

To establish our goal we have to prove two things: First, we have to show that using entanglement the cost function is bounded by some constant (see~\eqref{eq:constantcost}). 
For the quantum control strategy for the instance $(\mathcal{N}_t, P_X^{t}, k)$ we make some small modifications to the encoding and decoding strategies used by Cubitt et al. \cite{winter, winter2} for $\mathcal{N}$.  Roughly, the first controller will use the encoding strategy to modulate the signal on the line so as to encode the input $x$ for transmission via $\mathcal{N}_t$.  The second controller then uses the decoding strategy to determine $x$, which allows it to damp the signal down to 0.

Second, we have to show 
that by choosing $t$ large enough we can make the classical cost function exceed the quantum one (see~\eqref{eq:divergingcost}).
We will establish this result by contradiction. In particular, if, instead of diverging to infinity, the classical cost was uniformly bounded for all $t$, we show that the first controller's signal $c_1(x)$ is also uniformly bounded.  Based on this, since the input values $x$ grow without limit as $t$ increases, the second controller $c_2$ must do the bulk of the damping for large enough $t$.  In order to achieve this damping $c_2$ must have very good information about $x$, and for large enough $t$ we find that $c_2$ must have perfect information on $x$. 
This implies that $c_2$ can perfectly decode the outputs of $\mathcal{N}$ for all inputs $x$ just as the quantum controllers do. 
In turn, this allows us to construct a zero error coding scheme for $\mathcal{N}$ contradicting the classical zero-error capacity of $\mathcal{N}$.
\section{Preliminaries}\label{sec:prelim}
\subsection{Zero-error classical communication}
The scenario considered by Cubitt et al. involves a noisy classical channel $\mathcal{N}$ where the goal is to find an encoding scheme which allows the transmission of classical information over the noisy channel with \emph{zero error}.  $\mathcal{N}$ has input set $\mathcal{I}$ and output set $\mathcal{O}$ and the channel is given as $\mathcal{N}(o \, | \, i)$, the probability of obtaining output $o$ when the input is $i$.  In other words, $\mathcal{N}$ is a conditional probability distribution.
We wish to use $\mathcal{N}$ to send messages using some encoding scheme $\mathcal{E}$ and decoding scheme $\mathcal{D}$, as in Figure~\ref{fig:channelCom}.  The encoding and decoding scheme may make use of some shared resource $\rho$.

\begin{figure}[h!]
\includegraphics[scale=1.2]{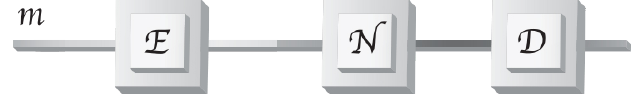}


\caption{Channel Communication with Correlation Assisted}
\label{fig:channelCom}
\end{figure}
When $\mathcal{M}$ is the input alphabet of $\mathcal{E}$ and $\mathcal{I}_{\mathcal{M}}$ is the identity channel on $\mathcal{M}$, then the \emph{zero-error capacity} of $\mathcal{N}$ when using the encoding/decoding strategy $(\mathcal{E}, \mathcal{D}, \rho)$ is 
\begin{equation}
	Z^\mathcal{N}(\mathcal{E}, \mathcal{D}, \rho) =
	\begin{cases}
	|\mathcal{M}| &\text{ if }\mathcal{D}(\mathcal{N}(\mathcal{E})) = \mathcal{I}_{\mathcal{M}}\\
	0 & \text{ otherwise}
\end{cases}
\end{equation}
Then the zero-error capacity of $\mathcal{N}$ using strategies from the set $\Omega$ is
\begin{equation}
Z^\mathcal{N}_\Omega  := \max_{(\mathcal{E}, \mathcal{D}, \rho) \in \Omega} 
	Z^\mathcal{N}(\mathcal{E}, \mathcal{D}, \rho)\ .
\end{equation}

As with Witsenhausen's problem, one can consider three classes of strategies: the deterministic strategies $D$, strategies using shared randomness $SR$, and strategies using shared entanglement $SE$.  Strategies in $D$ are  $(\mathcal{E}, \mathcal{D},0)$, where $\mathcal{E} : \mathcal{M} \rightarrow \mathcal{I}$, $\mathcal{D} : \mathcal{O} \rightarrow \mathcal{M}$ and $\rho$ can be taken to be 0.  Strategies in $SR$ are triples $(\mathcal{E}, \mathcal{D}, \rho)$ where $\rho$ is a probability distribution on some set $R$, $\mathcal{E} : \mathcal{M} \times R \rightarrow \mathcal{I}$ and $\mathcal{D} : \mathcal{O} \times R \rightarrow \mathcal{M}$.  

For a strategy $(\mathcal{E}, \mathcal{D}, \rho) \in SE$, 
$\rho \in \bop(\hil_\mathcal{E}\otimes \hil_\mathcal{D})$ is a quantum state shared between the encoder holding system $\hil_\mathcal{E}$ and the decoder holding system $\hil_\mathcal{D}$.
Without loss of generality, $\mathcal{E}$ performs a measurement on $\hil_\mathcal{E}$ indexed by its input $m$, and outputs the measurement outcome $i = \mathcal{E}(m)$ on the wire. 
Formally, this means that the encoding strategy
$\mathcal{E}$ is defined by measurement operators $\{\{\Pi_{m,i}^{1} \geq 0\}_{i}\}_m$ where for each $m$
we have $\sum_{i} \Pi_{m,i}^{1} = I_{\mathcal{H}_{\mathcal{E}}}$. 
Similarly, the decoder $\mathcal{D}$ performs measurements $\{\{\Pi_{o,m}^{2} \geq 0\}_m\}_{o}$ on $\hil_\mathcal{D}$ indexed by its input $o$ and outputs the outcome $m = \mathcal{D}(o)$. 

\subsection{Optimal deterministic encodings}
For deterministic encodings the input and output encodings are simply relabellings of some subsets of $\mathcal{I}$ and $\mathcal{O}$.  The critical consideration, then, is which subsets to use.  To this end we consider the notion of \emph{confusability}.  For element $i \in \mathcal{I}$ let $T_i = \{o \in \mathcal{O}\, |\, \mathcal{N}(o | i) > 0\}$ be the set of outputs to which $\mathcal{N}$ maps $i$. Two symbols $i, j \in \mathcal{I}$ are said to be {\itshape confusable} if their output distributions have at least one common element, i.e. if $T_i \cap T_j \neq \emptyset$.

The {\itshape confusability graph} $G(\mathcal{N})$ of the classical channel, $\mathcal{N}$, is a graph whose vertex set is $\mathcal{I}$ where two vertices are adjacent whenever they are confusable.  
\begin{eqnarray}
	V(G) & = & \mathcal{I} \\
	E(G) & = & \{(i,j)\, |\, T_{i} \cap T_{j} \neq \emptyset\}
\end{eqnarray}

Note that, for any zero-error deterministic encoding scheme, two confusable elements of $\mathcal{I}$ cannot be used to encode two different messages since it is impossible to perfectly distinguish the two elements based on the output of $\mathcal{N}$.  Hence we can put an upper bound on $Z_{D}^\mathcal{N}$ by determining the maximum size of a set of vertices in which no two are adjacent, i.e.\ the size of the largest independent set in $G(\mathcal{N})$, called the \emph{independence number} of $G(\mathcal{N})$.  

In fact, every independent set $S \in \mathcal{I}$ also gives a zero-error encoding strategy.  By definition the $T_i$'s are pairwise disjoint for $i,j \in S$ since all elements of $S$ are pairwise non-confusable.  We can then encode a message $m$ by mapping it to an element $i \in S$, and decode by mapping all elements of $T_i$ back to $m$.  So $Z_{D}^\mathcal{N}$ equals the independence number of $G(\mathcal{N})$ \cite{winter}.

\subsection{Channel Description and Classical Bound}
\label{sec:channeldescription}
In this subsection, we describe a channel $\mathcal{N}$ such that $Z_{SE}^{\mathcal{N}} > Z_{D}^\mathcal{N}$, due to Cubitt et al. \cite{winter}, \cite{winter2}).  We will use $\mathcal{N}$ in the construction of our Witsenhausen problem.

We will specify the channel by first specifying the confusability graph for our desired channel $\mathcal{N}$, which is defined in terms of a Kochen - Specker basis set.  We then construct the channel from this graph.

\begin{definition}
Let $U = \left\{B_{m} | m \in \indexset{q} \right\}$ be a set of $q$ orthonormal bases $B_{m} = \left\{\ket{b_{mj}} \, | \, j \in \indexset{d}\right\}$ for $\mathbb{C}^{d}$. 
If every set of vectors containing at least one vector from each $B_{m}$ contains at least one pair of orthogonal vectors then we call $U$ a \emph{Kochen - Specker basis set} (or a \emph{KS basis set}). 
\end{definition}
\noindent
Cubbit et al.\ use such a set of bases $U$ with $(q,d) = (6,4)$ in the construction of $\mathcal{N}$ \cite{winter}.

For our channel\footnote{
This channel is different from one used by Cubbit et al.\ \cite{winter} but it has the same confusability graph and deterministic capacity.  As well, the same quantum encoding/decoding strategy also works.
}
, set $\mathcal{I} = [q] \times [d]$ and $\mathcal{O}$ to the set of unordered pairs of $[q] \times [d]$, i.e., an element of $\mathcal{O}$ looks like $\{(m,j), (m^\prime, j^\prime)\}$ with $(m^\prime, j^\prime) \neq (m,j)$, corresponding to an edge. Set $\hat{T}_{(m,j)}$ to be the set of pairs $(m^\prime, j^\prime)$ such that $\braket{b_{mj}}{b_{m^\prime j^\prime}} = 0$.  The channel is defined by
\begin{equation}
	\mathcal{N}(\{(m,j), (m^\prime, j^\prime)\}  | (m,j)) = 
		\begin{cases}
			\frac{1}{|\hat{T}_{(m,j)}|} & o \in \hat{T}_{(m,j)} \\
			0 & o \notin \hat{T}_{(m,j)}.
		\end{cases}
\end{equation}
Note that if $i \notin o$, then $\mathcal{N}(o | i) = 0$.

For the bases $U$ with $(q,d) = (6,4)$ used by Cubbit et al.\ the degree of each vertex $(m,j)$ is $9 =|\hat{T}_{(m,j)}|$. Hence $\mathcal{N}(o | i)$ is either $1/9$ or 0.

Consider the confusability graph of $\mathcal{N}$ which has vertices $\mathcal{I}$ and edges which are unordered pairs of vertices, i.e.\ elements of $\mathcal{O}$.  Two elements $(m,j), $ and$(m^\prime, j^\prime)$ are confusabale whenever they can both be mapped to a common output. The only possible common output is $\{(m,j), (m^\prime, j^\prime)\}$.  Hence the edge set of $G(\mathcal{N})$ is just the elements of $\mathcal{O}$ which occur with non-zero probability. These are the $\{(m,j), (m^\prime, j^\prime)\}$ for which $\braket{b_{mj}}{b_{m^\prime j^\prime}} = 0$.

We now bound the classical zero-error capacity of $\mathcal{N}$~\cite{winter,winter2}. The set of vertices in $G(\mathcal{N})$ can be partitioned into $q$ cliques corresponding to the $q$ bases in $U$.  Any set of independent vertices can contain at most one vertex in each clique.  Further, since $U$ is a KS basis set, any set which contains a vertex in each of the $q$ cliques contains two vertices which are adjacent, and hence the set is not independent.  Thus the independence number of $G(\mathcal{N})$ is less than $q$ and $Z_{D}^\mathcal{N} = Z_{SR}^\mathcal{N} < q$.
 
\subsection{Quantum zero-error encoding for $\mathcal{N}$}
With an upper bound on the classical zero-error capacity of $\mathcal{N}$ in place, we recall the quantum strategy in $SE$ that beats the classical bound for zero error coding~\cite{winter}.

To begin, we first specify the entangled state that the encoder and decoder will share, which is the maximally entangled state $\ket{\psi} \in \mathcal{H}_\mathcal{E} \otimes \mathcal{H}_\mathcal{D}$, where $\mathcal{H}_\mathcal{E}$ and $\mathcal{H}_\mathcal{D}$ have dimension $d$.
\begin{equation}
	\ket{\psi} = \frac{1}{\sqrt{d}}\sum_{j \in \indexset{d}}\ket{j}\otimes\ket{j}
\end{equation}

When the encoder $\mathcal{E}$ receives a message $m \in \indexset{q}$ it measures its half of the shared state $\ket{\psi}$ along the basis $B_{m}^{*}$ (obtained by conjugating each state in $B_{m}$) to obtain outcome $j \in \indexset{d}$.  The residual state on the decoder's subsystem is then $\ket{b_{mj}}$.  The encoder then sends $(m,j)$ through the channel $\mathcal{N}$.

The decoder $\mathcal{D}$ receives from $\mathcal{N}$ some $s$ as its input, which is an edge in $G$ incident with $(m,j)$ and some other vertex $(m^\prime, j^\prime)$. 
The decoder thus cannot tell yet, whether the messages was $m$ or $m'$.
Since these vertices are adjacent, they correspond to orthogonal vectors $\ket{b_{mj}}$ (the residual state on $\mathcal{D}$'s subsystem) and $\ket{b_{m^\prime j^\prime}}$.  Hence upon receiving $s$, $\mathcal{D}$ measures its subsystem in some basis which includes $\ket{b_{mj}}$ and $\ket{b_{m^\prime j^\prime}}$.  The outcome is $(m,j)$ with probability 1.  $\mathcal{D}$ then outputs $m$.  (Note that $\mathcal{D}$ also determines $j$, which will be important later in our Witsenhausen construction.)

With $\mathcal{E}$ and $\mathcal{D}$ as above, and $\rho = \proj{\psi}{\psi}$, we find that 
\begin{equation}
Z_{SE}^\mathcal{N} \geq Z^\mathcal{N}(\mathcal{E}, \mathcal{D}, \rho) = q
\end{equation}
which beats the classical bound of $q-1$.

\section{Construction of the Witsenhausen problem}
\label{sec:witsenhausen}
\label{sec:WitProb}

For an instance of the Witsenhausen problem we need to give a channel, a probability distribution on $\mathbb{Z}$ for the inputs, and a constant $k$.  We now construct a family of channels $\mathcal{N}_t$ which we will use for our instances of the Witsenhausen problem.  We wish to use $\mathcal{N}: \indexset{q} \times \indexset{d} \rightarrow (\indexset{q}\times\indexset{d})^\times 2$
constructed in the previous section. It is clear that we cannot immediately use this channel since our wire yields inputs in $\mathbb{Z}$. We hence 
define a family of suitable encoding channels $\epsilon_t: \mathbb{Z} \rightarrow \indexset{q} \times \indexset{d}$ for $t \geq d$:
\begin{itemize}
	\item If $x = at + b$ for some $a \in \left[q\right]$ and $b \in \left[d\right]$ then $\epsilon_{t}$ outputs $(a,b)$.
	\item Otherwise $\epsilon_{t}$ outputs $(a, b)$ chosen uniformly at random from $\left[q\right] \times \left[d\right]$.
\end{itemize}

The channels we use in constructing our Witsenhausen problems are
\begin{equation}
\label{eq:ntDef}
	\mathcal{N}_{t} := \mathcal{N} \circ \epsilon_{t}
\end{equation}

We will also need to know the deterministic zero-error capacity of $\mathcal{N}_t$, given by
\begin{equation}
\label{contradict}
Z_{D}^{\mathcal{N}_{t}} = q - 1.
\end{equation}
To see this, note that the confusability graph $G^\prime$ of $\mathcal{N}_t$ has an induced subgraph isomorphic to $G$ (the confusability graph of $\mathcal{N}$) and all other vertices are all adjacent to every other vertex.  The latter vertices cannot be included in any independent set of size greater than 1, so the independence number cannot increase beyond that of $G$.  On the other hand, an independent set in $G$ maps to an independent set in $G^\prime$, so the independence number of $G^\prime$ is the same as that of $G$.

For the input probability distributions for our Witsenhausen problems, we define a distribution $P^{t}_{X}$ of the input $x$ by
\begin{equation}
\label{eq:distributionDef}
P_X^{t}(x) := \left\{ 
  \begin{array}{ll}
    \frac{1}{q} & \quad \text{if}\ x = mt \text{ where } m \in \indexset{q} \\
    0 & \quad \text{otherwise}
  \end{array} \right.
\end{equation}
Note that the value of $q$ goes hand in hand with the channel $\mathcal{N}$.
We have chosen this distribution since it ensures that no $x$ occurs for which $\epsilon_t$ generates a random output. In other words, with probability 1 the input to the Witsenhausen problem is mapped directly to a particular input to $\mathcal{N}$. Note that for a fixed $t$, the distribution~\eqref{eq:distributionDef} actually only depends on $m$, and hence we identify 
\begin{align}
P_M^{t}(m) := P_X^{t}(mt)\ .
\end{align}

Finally, we need to specify $k$.  Our results will work with any strictly positive $k$, but for concreteness one may take
\begin{equation}
	k = 1.
\end{equation}

\section{Quantum strategy for the Witsenhausen problem}
\label{sec:quantumstrategy}

Let us now show that there exists a quantum strategy using shared entanglement such that the cost function is a constant.
We will make use of the quantum encoding strategy for $\mathcal{N}$ proposed in~\cite{winter,winter2}, 
which achieves a zero-error capacity of $q$. The intuition is that the first controller, $c_1^{t}$, will perform the quantum encoding $\mathcal{E}$ for $\mathcal{N}_{t}$, suitably modified to modulate the signal so as to encode the message $m = x/t$ into $\mathcal{N}_t$.  Then $c_2^{t}$ will later decode according to the quantum decoding strategy $\mathcal{D}$ for $\mathcal{N}_{t}$ in order to determine $(m,j)$ and damp the signal on the line down to 0.

\begin{itemize}
	\item Controller $c_{1}^{t}$ does the following:
	\begin{itemize}
		\item Receive $x = mt$ and then apply $\epsilon_t(x) = (m,0)$ to find $m$.
		\item Measure $\ket{\psi}$ in basis $B_{m}^{*}$ to obtain $(m,j)$.
		\item Output $c_1^{t}(x) = j$ to be added to the signal which is then $y = x + j = mt+j$
	\end{itemize}
	\item Controller $c_{2}^{t}$ does the following:
	\begin{itemize}
		\item Receive output $s = \{(m,j), (m^{\prime}, j^{\prime})$ from the channel $\mathcal{N}_{t}$.
		\item Perform projective measurement on $\ket{\psi}$ on a basis containing $\ket{b_{m,j}}$ and $\ket{b_{m^{\prime}, j^{\prime}}}$ to obtain $(m,j)$.
		\item Output $c_2^{t}(s) = -mt - j$ to be added to the signal, which is then $z = y +  c_2^{t}(s) = 0$.
	\end{itemize}
\end{itemize}
Since final signal on the line is $z = mt + j - mt - j = 0$, the cost for this strategy is determined by the $k(c_{1}^{t}(x))^{2}$ term only:
\begin{align}
	 C^{\mathcal{N}_{t}, P_X^{t}, k}(c_{1}^t, c_{2}^t, \rho) 
	& =  k\underset{P_X^{t}, \mathcal{N}_{t}}{\mathbb{E}}\left[(c_{1}^t(mt))^{2}\right]\\
  & =  k\sum_{x}P_X^t(x)(c_{1}^t(x))^{2}\ .
\label{eq:quantumStep1} 
\end{align}
Recall that $P_{X}^{t}(x) = P_M(m)$ for $x = mt$ with $m \in [q]$ and 0 otherwise.  Hence
\begin{eqnarray}
	 C^{\mathcal{N}_{t}, P_X^{t}, k}(c_{1}^t, c_{2}^t, \rho) & 
   = & k\sum_{m \in \left[q\right]} P_M(m) (c_{1}^t(mt))^{2}\ .
\label{eq:quantumStep3}
\end{eqnarray}
Now since $0 \leq c_{1}^t(mt) \leq d - 1$
\begin{eqnarray}
C^{\mathcal{N}_{t}, P_X^{t}, k}(c_{1}^t, c_{2}^t, \rho) & 
   <  & k\sum_{m \in \left[q\right]} P_M(m) d^{2}	\\
  & = & kd^2\ .
  \label{step}
\end{eqnarray}
We thus have that $C^{\mathcal{N}_t, P_X^{t}, k}_{SE} < kd^2$ for all $t \geq d$ as promised in~\eqref{eq:constantcost}.

\section{Classical bound for the Witsenhausen problem}
\label{sec:classicalbound}

The challenging part of our result is to show that the classical cost function can be made to grow arbitrarily large by increasing $t$. 
Let us first give an outline of our proof, which can be divided into several natural steps. 
First we reduce to the case of deterministic strategies. Next we assume for contradiction that there is some uniform upper bound $M$ on the minimal deterministic cost for all $t$.  Taking some family of strategies $(c_{1}^{t}, c_{2}^{t}, 0)$ which achieve the minimal deterministic cost, we first show that $|c_{1}^{t}|$ and $|z|$ are also uniformly bounded.  Because of this, the signal after $c_{1}^{t}$ will be dominated by the $mt$ term for large $t$, i.e. $y \approx mt$.  Hence $c_{2}^{t}$ must have a very good estimate of $m$ in order to achieve a bounded final signal.  Put differently, for large enough $t$, $z \approx mt + c_{2}^{t} \approx 0$ so $-c_{2}^{t}/t \approx m$.  We use this fact to construct an encoding scheme for $\mathcal{N}_{t}$.  Roughly, the encoding scheme simulates the Witsenhausen problem to create $y$ which is sent through the channel.  The decoding scheme is to calculate $-c_{2}^{t}/t$ which will be equal to the input $m$.  For large enough $t$ this scheme has no error and beats the deterministic zero-error bound $Z_{D}^{\mathcal{N}_{t}}$, achieving our contradiction.

In the following, we will need to talk about distributions of $z$ and $s$. From their definitions, we see that these distributions do of course depend on the strategies of the controllers along with the distribution on $x$ and the channel $\mathcal{N}$.  Note, however, that $c_1^t$ and $c_{2}^{t}$ are deterministic functions of $x$  and $s$ and hence the probabilities of $z$ and $s$ are completely determined by $P_{X}^{t}$ and $\mathcal{N}_{t}$ once we fix $c_1^t$ and $c_{2}^{t}$.  More precisely, we have
\begin{eqnarray}
\label{eq:Pxt}
	P_{X,S}^{t}(x,s) & = & P_{X}^{t}(x)\mathcal{N}_{t}(s | x + c_{1}^{t}(x)) \\
\label{eq:Pzt}
	P_Z^{t}(z) & = & \sum_{\substack{x\\ z = x + c_1^t(x) + c_{2}^{t}(s)}}P_{X,S}^{t}(x,s)\ .
\end{eqnarray}
Furthermore, since all of our probability distributions have finite support (i.e. there is a finite output set for $\mathcal{N}$ and a finite set of possible inputs $x$) we can define a minimum probability:
\begin{eqnarray}
\label{eq:Pxtmin}
	P_{X}^{t\, min} & = & \min_{P_{X}^{t}(x) > 0} P_{X}^{t}(x) \\
\label{eq:Pztmin}
	P_{Z}^{t\, min} & = & \min_{P_{Z}^{t}(z) > 0} P_{Z}^{t}(z)
\end{eqnarray}

For a given $(q,d)$ we may calculate explicit values.  In particular, for $(q,d) = (6,4)$ and $\mathcal{N}^t$ and $P_X^t$ as given above, we find 
\begin{equation}
	P_X^{t\, min} = \frac{1}{6}.
\end{equation}
Although $P_Z^{t\, min}$ depends explicitly on $c_1^{t}$ and $c_2^{t}$ we can still find lower bounds on the minimal probability since these are deterministic functions.  Supposing that each $(x,s)$ gives a different value of $z$ so that there is no summing of probabilities, the minimal non-zero probability for a given $z$ is the same as the minimal probability for a given $(x,s)$.  In other cases the probability for particular values of $z$ can only go up.  Now, $s$ can depend implicitly on $x$ and $c_1^{t}$ via the input to $\mathcal{N}^t$, but for any input there 9 possible outputs from $\mathcal{N}^t$, each occurring with equal probability.  Thus we can lower bound the probability of any $(x,s)$ by  $1/6 \cdot 1/9$. So 
\begin{equation}
	P_Z^{t\, min}
	\geq 
	\frac{1}{6} \cdot \frac{1}{9} 
	= \frac{1}{54}.
\end{equation}

\subsection{Main theorem}

Here we state and prove the main theorem for the bound on $C_{SR}^{\mathcal{N}_{t}, P_X^{t}, k}$. For convenience sake, we first restate the theorem and provide some more details in the proof outline.
We will use various lemmas in the proof, which appear with proofs in later sections.

\begin{theorem}
	Let $(\mathcal{N}_{t}, P_X^{t}, k)$ be as given in Appendix \ref{sec:witsenhausen}.  Then for any $M \in \mathbb{R}$ there exists a $t_{0} \geq d$ such that for all $t \geq t_{0}$
	\begin{equation}
	C_{SR}^{\mathcal{N}_{t}, P_X^{t}, k} > M.
\end{equation}
\end{theorem}

\begin{proof}
We first show that we can restrict ourselves to deterministic strategies (Lemma~\ref{lemmaRandom}).  Hence it suffices to show that 
\begin{equation}
		C_{D}^{\mathcal{N}_{t}, P_X^{t}, k} > M
\end{equation}
for large enough $t$. To this end, let $(c_{1}^{t}, c_{2}^{t},0)$ be a family of deterministic strategies which achieve a cost of $C_{D}^{\mathcal{N}_{t}, P_X^{t}, k}$.  By assuming that 
\begin{equation}
\label{eq:boundforcontradictionMain}
		C_{D}^{\mathcal{N}_{t}, P_X^{t}, k} \leq M
\end{equation}
for all $t$, we will derive a contradiction. Using the definition of the cost function, this means that
\begin{equation}
	C_{D}^{\mathcal{N}_{t}, P_X^{t}, k}(c_{1}^{t}, c_{2}^{t},0) = 
	\mathbb{E}_{P_X,\mathcal{N}_t} \left[ 
		k c_{1}^{t}(x) + z^{2}
	\right] \leq M\ ,
\end{equation}
where $z = x + c_1^t(x) + c_{2}^{t}(s)$.

Since the expectation is taken over a finite set of events of non-zero probability, we can show that $|c_{1}^t(x)|$ and $|z|$ are uniformly bounded for all possible inputs $x$ (see Lemma~\ref{boundC1}). That is, 
\begin{eqnarray}
	|c_{1}^{t}(x)| & \leq& M_{X}\ , \\
	|z| & \leq & M_{Z}\ ,
\end{eqnarray}
for some $M_{X}$ and $M_{Z}$ which are independent of $t$ since they can be defined in terms of the parameters $M$, $k$, $P_{X}^{t\, min}$ and $P_{Z}^{t\, min}$ which are all independent of $t$. (The distributions $P_{X}^{t}$ and $P_{Z}^{t}$ depend on $t$, but the \emph{minimal} probability in these distributions does not.)  Then, for large enough $t$, $|c_{1}^{t}(x)|$ is small compared with the inputs which are of the form $x=mt$. 
Dividing by $t$ to make this more apparent, the signal $y = x + c_1^t(x)$ after the first controller satisfies
\begin{eqnarray}
	\frac{y}{t} & = &  m + \frac{c_{1}^{t}(mt)}{t} \\
	& \leq & m + \frac{M_{X}}{t} \\
	& \approx & m 
\end{eqnarray}
for large $t$ since $M_{X}/t \rightarrow 0$ as $t \rightarrow \infty$. Now the final output satisfies
\begin{eqnarray}
	\frac{z}{t} & = & m + \frac{c_{1}^{t}(mt)}{t} + \frac{c_{2}^{t}(s)}{t} \\
	& \approx &  m + \frac{c_{2}^{t}(s)}{t} \\
	& \approx & 0
\end{eqnarray}
for large $t$. The second line follows from the fact that $c_{1}^t(x)$ is uniformly bounded for all possible $x$ (Lemma~\ref{boundC1}), and the 
final line follows from the fact that also $|z|$ is uniformly bounded. Hence, we have that $z/t \rightarrow 0$ as $t \rightarrow \infty$. 
Lemma~\ref{lemma:c2approximatesmt} makes this approximation formal.  In particular, when $t_0 = 2(M_{X} + M_{Y})+ 1$, for any $t \geq t_{0}$ we have
\begin{align}\label{eq:estimate}
	\left|\frac{x}{t} + \frac{c_2^{t}(s)}{t} \right| < \frac{1}{2}
\end{align}
for all $(x,s)$ such that $P_{X,S}(x,s) > 0$ (in particular, for all $x = mt$ with $m \in [q]$).  

In our concrete example with $(q,d) = (6,4)$ and $k=1$ we can bound $M_X$ and $M_Z$ from above by $\sqrt{6M}$ and $\sqrt{54M}$, respectively.  In this case, $t_0$ is bounded above by $20\sqrt{M} + 1$.

From now on let $t \geq t_{0}$.  Note that~\eqref{eq:estimate} means that $c_2^{t}(s)/t$ forms a good estimate for $m = x/t$. This allows us 
to construct a zero error encoding scheme for the channel $\mathcal{N}_{t}$ as in Figure \ref{fig:channelCom}. 
We use messages set $\mathcal{M} = [q]$ and let 
\begin{equation}
	\eta = \frac{-c_2^{t}(s)}{t}\ .
\end{equation}
We then define the encoding scheme $\mathcal{E}$ and decoding scheme $\mathcal{D}$ as
\begin{itemize}
	\item $\mathcal{E}(m) = mt + c_{1}^{t}(mt)$.
	\item $\mathcal{D}(s)$ is given by rounding off $\eta$ to the nearest integer.
\end{itemize}
Since $|m - \eta| < 1/2$ by~\eqref{eq:estimateMain}, the nearest integer to $\eta$ is always $m$ and $\mathcal{D}$ always decodes correctly.  Here the fact that we have used a particular distribution for the inputs does not matter.  Since we achieve zero probability of error the encoding strategy must work with probability 1 for every input with positive probability.  Put differently, our encoding strategy works with zero error for an alphabet of $q$ symbols.  This contradicts the fact that only $q-1$ symbols can be sent over $\mathcal{N}_t$ with zero error.
\end{proof}

\subsection{Step 1: Deterministic strategies are optimal}
Here we reduce to the case of deterministic strategies.  Roughly, any strategy using shared randomness is a convex combination of deterministic strategies.  All the strategies in the convex combination have their costs bounded below by the deterministic cost, and the convex combination cannot take the cost any lower.

\begin{lemma}
\label{lemmaRandom}
Let $(\mathcal{N}, P_X, k)$ be an instance of the discrete Witsenhausen problem.  Then 
\begin{equation}
C_{D}^{\mathcal{N}, P_X, k} = C_{SR}^{\mathcal{N}, P_X, k}\ .
\end{equation}
\end{lemma}

\begin{proof}
Let $(c_{1}, c_{2}, \rho) \in SR$ be a strategy that achieves $C_{SR}^{\mathcal{N}, P_X, k}$.  Let us suppose that $\rho$ is drawn from a set $R$ and distributed according to some measure $d\rho$ with $\int_R d\rho = 1$.  The cost function \eqref{eq:costFunc} can then be written as
\begin{align}
  & C^{\mathcal{N}, P_X, k}(c_{1}, c_{2}, \rho) \\		
= & \underset{P_X,\, \mathcal{N}, d\rho}{\mathbb{E}}\left[kc_{1}^{2}\left(x, \rho\right) + (c_{2}(s, \rho) + x + c_{1}(x, \rho))^{2}\right] \\
 = & \int_R 
 \underset{P_X,\, \mathcal{N}}{\mathbb{E}}  
	\left[
		kc_{1}^{2}\left(x, \rho\right) + (c_{2}(s, \rho) + x + c_{1}(x, \rho))^{2}
	\right]  d\rho 
\label{lastStepDet}
\end{align}
For a fixed $\rho$, $(c_1(x,\rho), c_2(s,\rho), 0)$ is a deterministic strategy, and for each fixed $\rho$ the integrand in the above is the cost of that deterministic strategy, which is 
bounded below by the deterministic cost. Hence the above is bounded below by
\begin{eqnarray}
	C_{SR}^{\mathcal{N}, P_X, k} 
	& = &
		C^{\mathcal{N}, P_X, k}(c_{1}, c_{2}, \rho) \\
	& \geq &
		\int_R C_{D}^{\mathcal{N}, P_X, k} d\rho  \\
	& = &
		C_{D}^{\mathcal{N}, P_X, k}.
\end{eqnarray}
On the other hand, any deterministic strategy is also a shared randomness strategy in which $\rho$ is ignored, so 
\begin{equation}
C_{D}^{\mathcal{N}, P_X, k} \geq C_{SR}^{\mathcal{N}, P_X, k}\ .
\end{equation}
\end{proof}

\subsection{Step 2: Bounding controller 1 and the output}
Once we assume that there is some bound on the deterministic cost, then we can derive a bound on the output of $c_{1}^{t}$.  This follows straightforwardly from the fact that the cost is an expectation over some finite support, so each possible output from $c_{1}^{t}$ occurs with at least some minimal probability.  As well, $z$'s contribution to the cost is always positive, so we can omit it and the bound remains valid.  Then, multiplying $c_{1}^{t}$ by this minimal probability we must obtain some value which, when squared, is less than the bound.  A similar argument applies to $z$.

\begin{lemma}
\label{boundC1} Suppose $C^{\mathcal{N}_t, P_X^t, k}_{D} \leq M $.  Let $(c_{1}^t, c_{2}^t, 0) \in D$ be a deterministic strategy such 
that $C^{\mathcal{N}_t, P_X^t, k}_{D} = C^{\mathcal{N}_t, P_X^t, k}(c_{1}^{t}, c_{2}^{t}, 0)$. Then for all $x$ such that $P_{X}^{t}(x) > 0$  and $z$ such that $P_{Z}(z) >0$ 
\begin{eqnarray}
|c_{1}^{t}(x)| & \leq &  M_{X} = \sqrt{\frac{M}{k P_X^{t\, min}}} \\
|z| & \leq & M_{Z} = \sqrt{\frac{M}{P_{Z}^{t\, min}}}
\ .
\end{eqnarray}
\end{lemma}

\begin{proof}
By assumption, we have
\begin{align}
\label{eq:boundStep1}
M & \geq C^{\mathcal{N}_t, P_X^t, k}(c_{1}^t, c_{2}^t, 0) \ .
\end{align}
By definition and linearity of expectation,
\begin{align}
& C^{\mathcal{N}_{t}, P_X^t, k}(c_{1}^{t}, c_{2}^{t}, 0) \\	
&= \underset{P_X^t,\, \mathcal{N}_{t}}{\mathbb{E}}\left[k(c_{1}^{t}(x))^{2} + (c_{2}^t(s) + x + c_{1}^t(x))^{2}\right]\\
&= \underset{P_X^t,\, \mathcal{N}_{t}}{\mathbb{E}}\left[k(c_{1}^{t}(x))^{2}\right] + \underset{P_X^t,\, \mathcal{N}_{t}}{\mathbb{E}}\left[(c_{2}^t(s) 
+ x + c_{1}^t(x))^{2}\right]
\label{eq:lastStepDeter}
\end{align}
Since $(c_{2}^t(s) + x + c_{1}^t(x))^{2} \geq 0$, we also have that
$\underset{P_X^t,\, \mathcal{N}_t}{\mathbb{E}}\left[(c_{2}^t(s) + x + c_{1}^t(x))^{2}\right] \geq 0$. Therefore, we have from (\ref{eq:lastStepDeter}) that
\begin{equation}
C^{\mathcal{N}_t, P_X^t, k}(c_{1}^t, c_{2}^t, 0) \geq \underset{P_X^t,\, \mathcal{N}_t}{\mathbb{E}}\left[k(c_{1}^t(x))^{2}\right]\ .
\label{eq:boundStep3}
\end{equation}
Because $c_{1}^t$ does not depend on $\mathcal{N}_t$, we have
\begin{equation}
\underset{P_X^t,\, \mathcal{N}_t}{\mathbb{E}}\left[k(c_{1}^t(x))^{2}\right] = \underset{P_X^t}{\mathbb{E}}\left[k(c_{1}^t(x))^{2}\right]
\label{eq:boundStep2}
\end{equation}
Combining (\ref{eq:boundStep1}), (\ref{eq:boundStep2}) and (\ref{eq:boundStep3}) we thus have
\begin{align}
M &\geq \underset{P_X^t}{\mathbb{E}}\left[k(c_{1}^t(x))^{2}\right]\\
  	&= k\sum_{m \in \left[q\right]}P_X^t(x)(c_{1}^t(x))^{2}\\
	&\geq k\sum_{m \in \left[q\right]}P_X^{t\, min}(c_{1}^t(x))^{2}	
\label{lastStep}									
\end{align}
where $P_{X}^{t\, min}$ is defined in~\eqref{eq:Pxtmin}.  All the terms in the final sum are positive, hence for all $m$
\begin{align}
k P_X^{t\, min}((c_{1}^t(x))^{2} \leq M\ ,
\end{align}
and thus
\begin{align}
\left|c_{1}^t(x)\right| \leq \sqrt{\frac{M}{k P_X^{t\, min}}}\ .
\end{align}
By an analogous argument 
\begin{equation}
	|z| \leq \sqrt{\frac{M}{P_{Z}^{t\, min}}}
\end{equation}
where $P_{Z}^{t\, min}$ is defined in~\eqref{eq:Pztmin}.
\end{proof}

For the $\mathcal{N}^t$ and $P_X^t$ defined in appendix~\ref{sec:witsenhausen} with $(q,d)= (6,4)$ we can give explicit bounds on $M_X$ and $M_Z$.  Using $k=1$ we find
\begin{eqnarray}
	M_X 
	& \leq &
	\sqrt{
	 	6M
	 }
\\
	M_Z
	& \leq &
	\sqrt{
		54M
	}.
\end{eqnarray}

\subsection{Step 3: Approximating $m$}
With $c_{1}^t(x)$ and $z$ bounded, we find that $c_{2}^t$ must have a very good estimate on $x=mt$.  We divide out by $t$ to make this more apparent.

\begin{lemma}\label{lemma:c2approximatesmt}
Suppose $C^{\mathcal{N}_{t}, P_X^{t}, k}_{D} \leq M $ for all $t \geq d$. For each $t$ let $(c_{1}^{t}, c_{2}^{t}, 0) \in D$ be a deterministic strategy such that $C^{\mathcal{N}_{t}, P_X^{t}, k}_{D} = C^{\mathcal{N}_{t}, P_X^{t}, k}(c_{1}^{t}, c_{2}^{t}, 0)$.  Then there exists a $t_0 \geq d$ such that for all $t \geq t_0$
\begin{equation}
\left|\frac{x}{t} + \frac{c_2^{t}(s)}{t} \right| < \frac{1}{2}
\end{equation}
for all $(x,s)$ such that $P_{X,S}(x,s) > 0$. 
 
\end{lemma}

\begin{proof}
As shown in Lemma~\ref{boundC1}, $|c_{1}^{t}(x)|$ and $|z| = |x + c_{1}^{t}(x) + c_{2}^{t}(s)|$ are uniformly bounded,   
\begin{eqnarray}
	|c_{1}^{t}(x)| & \leq& M_{X} \\
	|x + c_{1}^{t}(x) + c_{2}^{t}(s)| & \leq & M_{Z}
\end{eqnarray}
with $M_{X}$ and $M_{Z}$ independent of $t$.  Combining these, we find
\begin{equation}
\left|
	x + c_2^{t}(s)
\right|
\leq 
M_{X} + M_{Y}
\end{equation}
Recall that $x = mt$ for some $m \in [q]$ for every $x$ such that $P_{X}^{t}(x) > 0$.  Dividing the above inequality by $t > 0$ to obtain
\begin{equation}
	\left|
		m + \frac{c_2^{t}(s)}{t}
	\right|
	\leq \frac{M_{X} + M_{Y}}{t}
\end{equation} 
for every $(m,s)$ with $m \in [q]$.  Then for $t \geq t_0 := 2(M_{X} + M_{Y})+ 1$, we have our desired result.
\end{proof}

We can make an estimate of $t_0$ for our explicit construction with $(q,d) = (6,4)$ and $k=1$.  Indeed
\begin{equation}
	t_0
	\leq
	2(\sqrt{54} + \sqrt{6})\sqrt{M} + 1
	\leq 20 \sqrt{M} + 1.
\end{equation}

\end{document}